\documentclass[11pt]{article}
\usepackage{amssymb,amsmath,multicol}
\usepackage{graphicx,url}
\usepackage{color,setspace,enumitem}
\usepackage{epstopdf}
\usepackage{caption}
\usepackage{times}
\usepackage[left=.9in, right=.9in, top=.9in, bottom=.9in]{geometry}

\usepackage[ruled]{algorithm}
\usepackage{algpseudocode}
\usepackage{ioa_code}

\usepackage{cite}

\newcommand{\case}[1]{
        {\vspace{1em}\noindent{\bf Case #1:}}}

\providecommand{\tup}[1]{%
    \relax\ifmmode
      \langle #1 \rangle%
    \else
        $\langle$#1$\rangle$%
    \fi
}

\newcommand{\act}[1]{%
    \relax\ifmmode
        \mathord{\mathcode`\-="702D\sf #1\mathcode`\-="2200}%
    \else
        $\mathord{\mathcode`\-="702D\sf #1\mathcode`\-="2200}$%
    \fi
}

\newcommand{\remove}[1]{}

\makeatletter
\def\mainlistofsymbols{
  \normalsize
  \vspace*{1.5 em}
  \@starttoc{los}
}

\def\partonelistofsymbols{
  \normalsize
  \vspace*{1.5 em}
  \@starttoc{p1los}
}

\def\parttwolistofsymbols{
  \normalsize
  \vspace*{1.5 em}
  \@starttoc{p2los}
}

\def\l@symbol#1#2{\addpenalty{-\@highpenalty} \vskip 4pt plus 2pt
{\@dottedtocline{0}{0em}{8em}{#1}{#2}}}
\makeatother

\newcommand{\newhiddensym}[2]{%
}

\newcommand{\stateSet}[1]{states(#1)}

\newcommand{\transSet}[1]{trans(#1)}

\newcommand{\algIOA}[2]{\ifmmode{\text{#1}_{#2}}\else{$\text{#1}_{#2}$}\fi}
\newcommand{\prIOA}[1]{A_{#1}}

\newcommand{\EX}{\ifmmode{\xi}\else{$\xi$}\fi}
\newcommand{\EXF}{\ifmmode{\phi}\else{$\phi$}\fi}

\newcommand{\extract}[2]{#1|#2}
\newcommand{\acts}{\alpha}

\newcommand{\hist}[1]{H_{#1}}

\newcommand{\st}{\sigma}

\newcommand{\objTSet}[1]{\Delta_{#1}}

\newcommand{\inter}[1]{
	\ifmmode{\left(\bigcap_{\mathcal{Q}\in#1}\mathcal{Q}\right)}
	\else{$\left(\bigcap_{\mathcal{Q}\in#1}\mathcal{Q}\right)$}
	\fi
}

\newcommand{\idSet}{\mathcal{I}}
\newcommand{\wSet}{\mathcal{W}}

\newcommand{\srvSet}{\mathcal{S}}

\newcommand{\objSet}{\mathcal{O}}
\newcommand{\verSet}{\mathit{Versions}}

\newcommand{\op}{\pi}
\newcommand{\trop}[2]{\op(#1)[#2]}

\newcommand{\rd}{\rho}
\newcommand{\wrt}{\omega}

\mathchardef\mhyphen="2D
\newcommand{\trw}[2]{cvr\mhyphen\wrt(#1)[#2]}
\newcommand{\trrd}[2]{cvr\mhyphen\rd(#1)[#2]}
\newcommand{\rrw}[2]{rr\mhyphen\wrt(#1)[#2]}

\newcommand{\pr}{p}

\newcommand{\wrtr}{w}

\newcommand{\bef}{\rightarrow}

\newcommand{\vid}[1]{\ifmmode{\nu_{#1}}\else{$\nu_{#1}$}\fi}

\newcommand{\seen}{\ifmmode{seen}\else{$seen$}\fi}

\newcommand{\ABD}{{\sc ABD}}
\newcommand{\mwABD}{{\sc mw}\ABD{}}

\newcommand{\valSet}{V}

\newcommand{\val}[1]{val_{#1}}

\newcommand{\tg}[1]{tag_{#1}}

\newcommand{\maxts}[1]{\ifmmode{maxTS_{#1}}\else{$maxTS_{#1}$}\fi}
\newcommand{\maxtag}[1]{\ifmmode{maxTag_{#1}}\else{$maxTag_{#1}$}\fi}
\newcommand{\maxpair}[1]{\ifmmode{maxMPair_{#1}}\else{$maxMPair_{#1}$}\fi}
\newcommand{\mintag}[1]{\ifmmode{minTag_{#1}}\else{$minTag_{#1}$}\fi}
\newcommand{\maxps}{\ifmmode{maxPS}\else{$maxPS$}\fi}
\newcommand{\conftg}[1]{\ifmmode{confirmed_{#1}}\else{$confirmed_{#1}$}\fi}
\newcommand{\maxconftag}{\ifmmode{\ms{maxCT}}\else{$maxCT$}\fi}
\newtheorem{theorem}{Theorem}
\newtheorem{lemma}[theorem]{Lemma}

\newtheorem{definition}[theorem]{Definition}

\newcommand{\prf}[1]{{}}

\newenvironment{proof}{\noindent{\bf Proof.}}{\hfill$\Box$\FF}

\newtheorem{Def}{Definition}[section]

\newcommand{\sacode}[5]
{ \vspace{.06in} \hrule \vspace{.06in} 
 \noindent {\bf #1}: \\
 \footnotesize \noindent {\bf Signature:}\B \nobreak
 \normalsize \begin{quote} \nobreak #2 \end{quote}
 \footnotesize \noindent {\bf States:}\B \nobreak
 \begin{quote} \nobreak #3 \end{quote}
 \noindent {\bf Transitions:} \nobreak
 \vspace{-.2in} \nobreak
 \normalsize #4
 \vspace{-.06in} \hrule \vspace{.06in} 
}

\newcommand{\seq}[1]{%
    \relax\ifmmode
      \langle \! \langle #1 \rangle \! \rangle%
    \else
        $\langle \! \langle$ #1 $\rangle \! \rangle$%
    \fi
}

\newcommand{\B}{\vspace*{-\smallskipamount}}

\newcommand{\FF}{\vspace*{\medskipamount}}

\newcommand{\Nat}{{\N}}
\newcommand{\N}{\mathbb N}

\newcommand{\ms}[1]{%
    \relax\ifmmode
        \mathord{\mathcode`\-="702D\it #1\mathcode`\-="2200}%
    \else
        {\it #1}%
    \fi
}

\newcommand{\lit}[1]{%
    \relax\ifmmode
        \mathord{\mathcode`\-="702D\sf #1\mathcode`\-="2200}%
    \else
        {\it #1}%
    \fi
}

\newcommand{\XDK}[1]{}
\newcommand{\uselater}[1]{} 

\renewcommand{\setminus}{-}

\newcommand{\cg}[1]{#1}

\newcommand{\nn}[1]{#1}
\newcommand{\af}[1]{#1}

\begin{document}


\title{
CoVer-ability: \textbf{Co}nsistent \textbf{Ver}sioning for Concurrent Objects 
\footnote{Supported in part by FP7-PEOPLE-2013-IEF grant ATOMICDFS No:629088, Ministerio de Economia y Competitividad grant TEC2014- 55713-R, Regional Government of Madrid (CM) grant Cloud4BigData (S2013/ICE-2894, co- funded by FSE \& FEDER), NSF of China grant 61520106005, and European Commission H2020 grants ReCred and NOTRE.}
}
%

\author{
Nicolas Nicolaou$^{~\dag}$ \and Antonio Fern\'andez Anta
\thanks{IMDEA Networks Institute, 
	Madrid, Spain, 
	\texttt{nicolas.nicolaou@imdea.org},\texttt{antonio.fernandez@imdea.org}}
\and Chryssis Georgiou
\thanks{Dept. of Computer Science, University of Cyprus,
  Nicosia, Cyprus, 
 \texttt{chryssis@cs.ucy.ac.cy}
 }
}


\maketitle 

\begin{abstract}


An \textit{object type} characterizes the domain space and the operations
that can be invoked on an object of that type. In this paper we introduce a new
property for concurrent objects, we call \emph{coverability}, that aims to 
provide precise guarantees on the consistent evolution of an object.   
%
This new property is suitable for a variety of distributed objects
including \emph{concurrent file objects} that demand operations 
to manipulate the latest version of the object. 
We propose two levels of coverability: (i) strong coverability and 
(ii) weak coverability. Strong coverability requires that only
a \emph{single operation} can modify \af{the latest} version of the object, i.e. \textit{``covers"}
the latest version with a new version, imposing 
a total order on object modifications.  
Weak coverability relaxes the strong 
requirements of strong coverability and allows \emph{multiple operations} 
to modify the same version of an object, where each modification leads to a different version. 
Weak coverability preserves consistent evolution of the object, by demanding any subsequent operation
to only modify one of the newly introduced versions. \cg{Coverability combined with atomic guarantees 
yield to \emph{coverable atomic read/write registers}.}
We also show that 
strongly coverable atomic registers are equivalent in power to consensus. 
Thus, we focus on {\em weakly} coverable registers, and we demonstrate 
their importance by showing that they cannot be implemented using similar types of registers,
like ranked-registers.
Furthermore we show that  
weakly coverable registers may be used to 
implement basic (weak) read-modify-write and file objects. 
Finally, we 
implement weakly coverable registers by modifying an existing MWMR atomic register
implementation.\vspace{2em}

\noindent{\bf Submission Type:} Regular paper.

\end{abstract}

%
%
%

\thispagestyle{empty}
\setcounter{page}{0}
\newpage

\section{Introduction}
\label{sec:intro}

\noindent{\textbf{Motivation and Prior Work.}} A concurrent system allows multiple processes to interact with a single object 
at the same time. A long string of research work \cite{Lamport79, Lamport86, HW90, GD05, AH07}  
has been dedicated to explain the behavior of concurrent objects,
defining the order and the outcomes of operations when those are invoked concurrently on the object. Lamport in \cite{Lamport79, Lamport86} presented three different incremental semantics, \emph{safety}, \emph{regularity}, and \emph{atomicity} that characterize the behavior of read/write objects (registers) when those are modified or read concurrently by multiple processes. The strongest, and most difficult to provide 
in a distributed system, is \emph{atomicity} which provides the illusion
that the register is accessed sequentially. Herlihy and Wing presented \emph{linearizability} in \cite{HW90}, an extension
of atomicity to general concurrent objects. More recent developments 
have proposed abortable operations in the event of concurrency \cite{AH07}, and ranked registers \cite{GD05} that allow operations to 
abort in case a higher ``ranked'' operation  was previously
or concurrently executed in the system. 

With the advent of cloud computing, emerging families of more 
complex concurrent objects, like files, distributed databases, and bulleting boards, demand precise
guarantees on the consistent evolution of the object.
\nn{For example, in \emph{concurrent file objects} one would expect 
that if a write operation $\wrt_2$ is invoked after a write operation $\wrt_1$ 
is completed, then $\wrt_2$ modifies either the version of the file written by $\wrt_1$ or a 
version of the file newer than the one written by $\wrt_1$.}
\emph{So is it possible to provide such guarantees using simpler 
objects as building blocks?}

In existing atomic read/write 
distributed shared register implementations,
write operations 
are usually \nn{allowed to modify the value of the register,
even when they are unaware of the value written by the latest preceding write operation.} 
In systems that assume 
a single writer \cite{ABD96, CDGL04, GNS09, GNS08}, the problem may be diminished by having 
the sole writer compute the next value to be written 
in relation to the previous values it wrote. 
The problem becomes more apparent when multiple writers may alter the value of a single register concurrently \cite{LS97, EGMNS09}.
In such cases, atomic read/write register implementations 
appear unsuitable to directly implement objects that 
demand evolution guarantees.  
%
%
Closer candidates to build such objects are the bounded \cite{BDFG03}
and ranked \cite{GD05} registers. These objects take into account the 
``rank" or sequence number of previous operations to decide whether
to allow a read/write operation to commit or abort. 
\nn{These approaches do not prevent, however, the use of an arbitrarily 
higher rank, and thus an arbitrarily higher version, than the previous operations. 
This affects the consistent evolution of the object, as intermediate versions 
of the object maybe ignored.}\vspace{.4em} 

\noindent{\bf Contributions.}
In this paper we propose a formalism to extend a concurrent object in such a way that the evolution of its state satisfies certain guarantees. To this end, we
extend an object state with a \emph{version,} and introduce the concept of \emph{coverability,} that defines how the versions of an object can evolve (Section~\ref{sec:atomic}). 

\nn{In particular, we first introduce 
a new class of a concurrent read/write register type, which we call \emph{versioned register}.
A concurrent register is of a \textit{versioned} type, if the state of the register and any operation (read or write) 
that attempts to modify the state of the register, are associated with a \emph{version}.
An operation may modify the state and the version of the 
register, or it may just retrieve its state-version pair.}

\textit{Coverability} defines the exact guarantees that a versioned register 
provides when it is accessed concurrently by multiple processes with respect to the evolution of its versions.
\nn{We define two levels of coverability: \emph{strong} and \emph{weak} coverability. 
Strong coverability ensures that only a \emph{single operation} may change a given version (and thus the state) of the register, 
resulting in a lineal evolution of the versions (and the states) of the register. 
Weak coverability relaxes this rule and allows \textit{multiple operations} to change a version, 
generating in this way a \emph{tree} with possibly multiple version branches that can grow in parallel. 
This shares similarities with \emph{fork linearizability} presented in \cite{MS02}. However, in contrast to \cite{MS02},
weak coverability allows processes, that change the same version of the object, to see the changes of each other in subsequent
operations. In particular, by weak coverability, when all the operations that extend a particular version of the object terminate,
there is one version $ver$ that was generated by one of those operations, which is the ancestor of 
any version extended by any subsequent operation. 
Thus, only a single branch in the tree is extended and  
that branch denotes the evolution of the register. 
}
\cg{Combining strong/weak coverability with atomic guarantees we obtain {\em strongly/weakly coverable atomic read/write registers}. } 
While strongly coverable atomic registers are very desirable objects, we show that they are in fact very strong. In particular, 
we argue that these object types are as powerful as consensus objects (the details are given in Appendix \ref{appx:consensus}). 
Hence, it is challenging to implement these objects in some distributed systems, and impossible in an asynchronous system prone to failures (from the FLP 
result~\cite{FLP85}).

The good news is that even {\em weakly} coverable atomic registers have very interesting features. 
On the one hand, they can be implemented in message passing asynchronous distributed systems where processes can fail. 
To show this, we describe how algorithms that implement atomic R/W registers can easily be modified to implement these objects (Section \ref{sec:algorithms}). 
On the other hand, we show that weakly coverable atomic registers cannot be implemented using other previously defined register types such as 
ranked registers (Section~\ref{ssec:tr-vs-rr}). 

One of the main motivation for introducing coverable registers are {\em file objects}, which can be seen as a special case 
of register objects in which each new value is a revision of the previous value. In essence, each modification of a file can be seen as an atomic read-modify-write (RMW) operation. Strongly coverable atomic registers provide the desired strong guarantees for files, since they are powerful enough to support atomic RMW operations. 
However, we show that even {\em weakly} coverable atomic registers can be used to provide interesting weak RMW guarantees that can be used to implement files with a good level of consistency 
 (Section~\ref{sec:applications}).

\section{Model}
\label{sec:model}

We consider a distributed system composed of $n$ \emph{asynchronous} processes,
with identifiers from a set $\idSet=\{p_1,\ldots,p_n\}$, 
each of which represents a sequential thread of control. Processes
may interact with a set of shared objects $\objSet$. 
Each object in $\objSet$ represents a \emph{data structure}
shared among the processes, and has a \emph{type} which defines the
possible set of \emph{object states} and the set of \emph{operations} 
that provide the means to manipulate the object. 
A subset of processes may fail by \emph{crashing}. 

Processes can be modeled in terms of I/O Automata \cite{LT89}.
An automaton $A$  (which combines the automata $A_i$ for each process $\pr_i\in\idSet$) is defined over a set of \emph{states}
and a set of \emph{actions}. 
An \emph{execution} $\EX$ of $A$ is an alternating sequence
of \emph{states} and \emph{actions} of $A$.
An \emph{execution fragment} is a finite prefix of an execution.
We say that an execution fragment
$\EX'$ \emph{extends} an execution fragment $\EX$,
if $\EX$ is a prefix of $\EX'$.
%
A \emph{history} of an automaton $A$, denoted by $\hist{\EX}$, is the subsequence 
of actions occurring in some execution fragment $\EX$. 
%
\remove{
An automaton $A$ can be constructed from the composition of 
compatible automata $\prIOA{i}$, one for each process $\pr_i\in\idSet$.
Each state $\st\in\stateSet{A}$ is a vector of the states of the 
component automata $\prIOA{i}$. Let $\st[i]$ denote the state of automaton
$\prIOA{i}$ in $\st$. For a step $\tup{\st,\acts,\st'}\in\transSet{A}$, the 
state of an automaton $\prIOA{i}$ moves from $\st[i]$ to $\st'[i]$ if 
$\tup{\st[i],\acts,\st'[i]}\in\transSet{\prIOA{i}}$; otherwise $\st[i]=\st'[i]$.
Therefore, we can extract the execution of an automaton $\prIOA{i}$ from
an execution $\EX$ of $A$, denoted by $\extract{\EX}{\prIOA{i}}$, by: (i) deleting
any steps $\tup{\st_k,\acts_k,\st_{k+1}}$ s.t. $\tup{\st_k[i],\acts_k,\st_{k+1}[i]}\notin\transSet{\prIOA{i}}$,
and (ii) by replacing the remaining $\st_z$, for $z\neq k$, with $\st_z[i]$ in $\EX$.
Similarly if $\hist{\EX}$ is the history of events for an execution $\EX$ of $A$,
then $\extract{\hist{\EX}}{\prIOA{i}}$ (or $\hist{\EX}^i$ for short), is the 
history of actions of $\prIOA{i}$ in $\EX$, and contains the subsequence 
of actions occurring in $\extract{\EX}{\prIOA{i}}$. 
}
An automaton $\prIOA{}$ 
\emph{invokes} an operation when an \emph{invocation
action} occurs in an execution $\EX$, and receives a \emph{response}
to an action when a \emph{response action} occurs.  
An operation $\op$ 
is \emph{complete} in an execution $\EX$, if $\hist{\EX}$ contains 
both the invocation and the matching response actions for $\op$; otherwise 
$\op$ is \emph{incomplete}. 
A history $\hist{\EX}$ of the automaton $A_i$ of a process $\pr_i$ is \emph{well formed} if it begins with an invocation 
event 
and alternates between matching invocation and response events. (This demonstrates the assumption that each process is a single thread of control.)
Each history $\hist{\EX}$ includes a \nn{precedence relation} $\bef_{\hist{\EX}}$ on its operations.
An operation $\op_1$ \emph{precedes} an operation $\op_2$ (or 
$\op_2$ \emph{succeeds} $\op_1$) in $\hist{\EX}$ if the response of 
$\op_1$ appears before the invocation of $\op_2$ in $\hist{\EX}$. This is 
denoted by $\op_1\bef_{\hist{\EX}}\op_2$. If $\op_1\not\bef_{\hist{\EX}}\op_2$ and
$\op_2\not\bef_{\hist{\EX}}\op_1$ in $\hist{\EX}$, then $\op_1$ and 
$\op_2$ are \emph{concurrent}. 
%
%
A process $\pr_i$ \emph{crashes} in an execution $\EX$ if the event $\act{fail}_{\pr_i}$ 
appears and is the last action of $\pr_i$ in $\hist{\EX}$; otherwise $\pr_i$ is \emph{correct}.



\section{Coverable Atomic Read/Write Registers}
\label{sec:atomic}


In this section we define a new type of R/W register, the \emph{versioned register}.
Next we provide new consistency properties for concurrent versioned registers 
called \emph{(strong/weak) coverability}.  We show how coverability can be combined 
with atomic guarantees to yield a coverable atomic register.\vspace{-0.5em}
%


\paragraph{Versioned register.} 
Let $\verSet$ be a \textit{totally ordered} set of \textit{versions}.
A \emph{versioned register} is a type of read/write register where each value written 
is assigned with a version from the set $\verSet$. 
Moreover, each write operation $\op$ that attempts to change the value of the 
register is also associated with a version, say $ver_\op$, denoting that it intends
to overwrite the value of the register associated with the version $ver_\op$. 
More precisely, an implementation of a R/W register 
offers two operations: \emph{read} and \emph{write}. 
A process $\pr_i\in\idSet$ \emph{invokes} a \emph{write} (resp. \emph{read}) operation when it issues 
a $\act{write}(\val{})_{\pr_i}$ (resp. $\act{read}_{\pr_i}$) request. 
The \emph{versioned} variant of a R/W register also offers two operations:
 (i) $\act{cvr-write}(\val{}, ver)_{\pr_i}$, and (ii) $\act{cvr-read}()_{\pr_i}$. 
A process $\pr_i$ invokes a $\act{cvr-write}(\val{}, ver)_{\pr_i}$ operation 
when it performs a write operation that attempts to change the value 
of the object. 
The operation returns
the value of the object and its associated version, along with a flag informing
whether the operation has successfully changed the value of the object or failed.
We say that a write is \emph{successful} if it changes the value of the 
register; otherwise the write is \emph{unsuccessful}.
The read operation $\act{cvr-read}()_{\pr_i}$ involves a request to 
retrieve the value of the object. The response of this operation is the 
value of the register together with the version of the object that this value is 
associated with. 

Read operations do not incur any change on the value 
of the register, whereas write operations attempt to modify the value of the register.
More formally, let $\objTSet{T}$ be the set of transitions for the versioned register. 
Then, each $\delta\in\objTSet{T}$ is a tuple 
 $\tup{\st, \op, \pr_i, \st', res}$, denoting
 that the register moves from state $\st$ to state $\st'$, 
  and responds with $res$, as a result of operation  
 $\op$ invoked by process $\pr_i\in\idSet$.
 The state of a versioned register is essentially its \emph{value}, drawn from a set $\valSet$,
 and its \emph{version}, drawn from the set $\verSet$.
 We assume that $\objTSet{T}$ is \emph{total}, that is,
 for every $\op\in\{\act{cvr-write}(\val{}, ver)_{\pr_i}, \act{cvr-read}()_{\pr_i}\}$, $\pr_i\in\idSet$, and 
 $\st = (\val{},ver) \in\valSet \times \verSet$, there exists $\st' = (\val',ver') \in\valSet \times \verSet$
 and $res$ such that $\tup{\st, \op, \pr_i, \st', res}\in\objTSet{T}$. 
As such, the transitions of the versioned register type 
can be written as follows:\vspace{-.5em} 
\begin{enumerate}[leftmargin=5mm]\itemsep2pt		
		\item $\tup{(\val{},ver), \act{cvr-write}(\val{}',ver_\wrt), \pr_i, (\val{}',ver'), (\val{}',ver', chg)}$, for $ver_\wrt = ver$,
		\item $\tup{(\val{},ver), \act{cvr-write}(\val{}',ver_\wrt), \pr_i, (\val{},ver), (\val{},ver, unchg)}$, for $ver_\wrt \neq ver$
		\item $\tup{(\val{},ver), \act{cvr-read}(), \pr_i, (\val{},ver), (\val{},ver)}$.
	\end{enumerate}
Notice that write operations may or may not modify the value/version 
of the register.
In the transitions above, $ver_\wrt$ denotes the version of the register which the write operation tries
to modify. 
The relationship of $ver$ with $ver'$ may vary depending on the application that uses this register (but seems
natural to assume that $ver' > ver$).
A read operation 
does not make any changes on the value or the version of the object. 
To simplify notation, 
in the rest of the paper 
we avoid any reference 
to the value of the register. Additionally we only use the flag when its value is $unchg$.
Thus, $\act{cvr-write}(v, ver)(v, ver', chg)_{\pr_i}$ is denoted as $\trw{ver}{ver'}_{\pr_i}$,
and $\act{cvr-write}(v, ver)(v', ver', unchg)_{\pr_i}$ is denoted as $\trw{ver}{ver', unchg}_{\pr_i}$. 

%
We say that, a write operation
 \emph{revises} a version $ver$ of 
the versioned register to a version $ver'$ (or \emph{produces} $ver'$) in an execution $\EX$, 
if $\trw{ver}{ver'}_{\pr_i}$ completes in $\hist{\EX}$.
Let the set of \emph{successful write} operations on a 
history $\hist{\EX}$ 
be defined as: 
\[
\wSet_{\EX,succ} = \{\op: \op=\trw{ver}{ver'}_{\pr_i}\text{ completes in }   \hist{\EX}\}
\] 
The set now of produced versions 
in the history $\hist{\EX}$ is defined by:
\[
\verSet_{\EX} = \{ver_i : \trw{ver}{ver_i}_{\pr_i}\in \wSet_{\EX, succ}\}\cup\{ver_0\}
\]
where $ver_0$ is the initial version of the object.
Observe that the elements of $\verSet_{\EX}$ are totally ordered.
In the rest of the text we use `$*$' in the place of some parameter
to denote that any legal value for that parameter can be used.
Now we present the {\em validity} property which defines explicitly the 
set of executions that are considered to be valid executions.

\begin{definition} [Validity]
\label{def:validity} 
An execution $\EX$ (resp. its history $\hist{\EX}$) is a \emph{valid execution} (resp. history)
on a versioned object, if  $\wSet_{\EX}$ 
and for any $\pr_i,\pr_j\in\idSet$:\vspace{-0.2em}
\begin{itemize}[leftmargin=5mm]\itemsep2pt
	\item $\forall \trw{ver}{ver'}_{\pr_i} \in \wSet_{\EX,succ}, ver < ver'$,
	\item for any operations $\trw{*}{ver'}_{\pr_i}$ and $\trw{*}{ver''}_{\pr_j}$ in $\wSet_{\EX,succ}$, $ver'\neq ver''$,  and
	\item  for each $ver_k\in Versions_{\EX}$ there is a sequence of versions $ver_0, ver_1,\ldots, ver_k$, such that 
	$\trw{ver_i}{ver_{i+1}}$ $\in\wSet_{\EX,succ}$, for $0\leq i<k$.
%
%
\end{itemize}
\end{definition}
Validity makes it clear that an operation changes the 
version of the object to a larger version, according to the total ordering of the versions.
Also validity specifies that versions are \emph{unique}, 
i.e. no two operations associate two states with the same version. This can be easily 
achieved by, for example, recording a counter and the id of the invoking 
process in the version of the object. Finally, validity requires that each 
version we reach in an execution is \textit{derived} (through a chain of operations) 
from the initial version of the register $ver_0$.
From this point onward we fix $\EX$ to be a valid 
execution and $\hist{\EX}$ to be its valid history.
\vspace{-.5em}

\paragraph{Coverability.} 
\sloppypar{%
We can now 
define the \emph{strong} and \emph{weak coverability} properties 
over a valid execution $\EX$ of versioned registers with respect to some 
total order $>_\EX$ on the operations of $\EX$. 
} 
%

%
%
\vspace{-.2em}

\begin{definition}[Strong Coverability]
\label{def:strong}
Let $ver_0 < ver_1 < \ldots < ver_{|\wSet_{\EX,succ}|}$ be the versions in $\verSet_{\EX}$.

A valid execution $\EX$ is \textbf{strongly coverable} with respect to a total order $<_{\EX}$ 
on operations in $\wSet_{\EX,succ}$ if:
\begin{itemize}[leftmargin=5mm]\itemsep2pt
	\item $\trw{ver_{i-1}}{ver_i}\in\wSet_{\EX,succ}$, for $1\leq i \leq |\wSet_{\EX,succ}|$, 
	\item  
	$\trw{ver_{i-1}}{ver_i}<_{\EX}\trw{ver_i}{ver_{i+1}}$, for $1\leq i < |\wSet_{\EX,succ}|$ , and 
	\item  
	if $\op_1, \op_2\in\wSet_{\EX,succ}$, and $\op_1\bef_{\hist{\EX}}\op_2$ then $\op_1<_{\EX}\op_2$. 
\end{itemize}\vspace{-.3em} 
\end{definition}
By Definition \ref{def:strong}, all successful write operations are totally ordered 
with respect to the versions they modify. Notice than only a single write 
operation modifies each version $ver_{i-1}$ to the next version $ver_i$.
Thus, strong coverability defines an object type which is 
difficult to provide in an asynchronous distributed setting. 
In fact it can be shown that strongly coverable
registers can be used to solve consensus among asynchronous 
fail-prone processes (see Appendix \ref{appx:consensus}). %
However, as shown by Fischer, Lynch and Paterson \cite{FLP85}, 
solving consensus in such a system is impossible in the existence of a single crash failure, unless
some powerful object is used.
Hence the interest in defining a {\em weaker} version of coverability.

\begin{definition}[Weak Coverability]
\label{def:weak}
A valid execution $\EX$ is \textbf{weakly coverable} with respect to a total order $<_{\EX}$ 
on operations in $\wSet_{\EX,succ}$ if:
\begin{itemize}[leftmargin=5mm]\itemsep2pt
	\item ({\bf Consolidation}) If $\op_1=\trw{*}{ver_i}, \op_2=\trw{ver_j}{*} \in \wSet_{\EX,succ}$, 
	and $\op_1\bef_{\hist{\EX}} \op_2$ in $\hist{\EX}$, then $ver_i \leq ver_j$ and $\op_1<_{\EX}\op_2$.
	\item ({\bf Continuity})  if $\op_2=\trw{ver}{ver_i}\in\wSet_{\EX, succ}$, then 
	there exists $\op_1\in\wSet_{\EX,succ}$ s.t. 
	$\op_1=\trw{*}{ver}$ and $\op_1<_\EX \op_2$, or $ver=ver_0$.
	\item ({\bf Evolution})
	let $ver, ver', ver''\in Versions_{\EX}$. If there are sequences of versions $ver'_1, ver'_2,\ldots, ver'_{k}$
	and $ver''_1, ver''_2,\ldots, ver''_{\ell}$, where $ver=ver'_1=ver''_1$, $ver'_{k}=ver'$, and $ver''_{\ell}=ver''$
	such that
	$\trw{ver'_i}{ver'_{i+1}}$ $\in\wSet_{\EX,succ}$, for $1\leq i<k$, and $\trw{ver''_i}{ver''_{i+1}}$ $\in\wSet_{\EX,succ}$, for $1\leq i<\ell$,
	and $k < \ell$, then $ver' < ver''$.
\end{itemize}
\end{definition}
By Definition \ref{def:weak}, weak coverability allows multiple write operations to revise the same version $ver_i$
of the register, each to a \emph{unique} version $ver_j$. 
%
Given the set of successful operations $\wSet_{\EX,succ}$
and the set of versions $\verSet_{\EX}$, Definitions \ref{def:validity} and \ref{def:weak} define a connected rooted tree $\mathcal{T}$ s.t.:
\begin{itemize}[leftmargin=5mm]\itemsep2pt
	\item The set of nodes of $\mathcal{T}$ is $\verSet_{\EX}$,
	\item \af{The initial version $ver_0$ of the object} is the root of $\mathcal{T}$,
	\item A node $ver_i$ is the parent of a node $ver_j$ in $\mathcal{T}$ iff $\exists\trop{ver_i}{ver_j}\in\wSet_{\EX,succ}$, 
	\item If $\op_1=\trw{*}{ver_i}\in \wSet_{\EX,succ}$, 
		s.t. $\op_1$ is not concurrent with any other operation, then 
		$\forall \op_2\in\wSet_{\EX,succ}$,
		s.t. $\op_1\bef_{\EX} \op_2$ and $\op_2=\trop{ver_z}{*}$, then $ver_i$ is an ancestor of $ver_z$ in $\mathcal{T}$,
		\af{or $ver_i=ver_z$ (by Consolidation, Continuity, and Validity)}
	\item if $ver_i$ is an ancestor of $ver_j$ in $\mathcal{T}$, then $\trw{*}{ver_i}<_\EX\trw{*}{ver_j}$ (by Continuity).
	\item if $ver_i$ is at level $k$ of $\mathcal{T}$ and $ver_j$ is at level $\ell$ of $\mathcal{T}$ s.t. $k<\ell$, then
	$ver_i<ver_j$ (by Evolution).
\end{itemize}
Observe that without the properties imposed by weak coverability, 
validity allows the creation of a tree of versions and 
does not prevent operations from being applied on an old version of the register. 
\emph{Continuity}, \emph{Consolidation}, and \emph{Evolution} explicitly specify the conditions that reduce the 
branching of the generated tree, and in the case of not concurrency lead the 
operations to a single path on this tree. 
%
\emph{Consolidation} specifies that write operations
may revise the register with a version larger than any version modified
 by a preceding write operation, and may lead to a version newer than
 any version introduced by a preceding write operation. 
\emph{Continuity} defines that a write operation may revise a version that was introduced
by a preceding write operation according to the given total order.
Finally, \emph{Evolution} limits the relative increment on the version of a register that can be
introduced by any operation.
Figure \ref{fig:tree} provides an illustration of a tree 
created from a coverable execution $\EX$. We box 
sample instances of the execution and we indicate the coverability properties they satisfy.
\vspace{-.5em}

\paragraph{Atomic coverability.} 
We now combine coverability with atomic guarantees to obtain coverable atomic read/write registers. 
A register is linearizable \cite{HW90}, or equivalently \emph{atomic} (as 
defined specifically for registers by \cite{Lynch1996, Lamport86}) if the following conditions are 
satisfied by any execution $\EX$ of an implementation of the object.\vspace{-.2em}  


\begin{definition}[Atomicity]
\label{def:atomic}
{\rm  \cite[Section 13.4]{Lynch1996} An execution $\EX$ of an automaton $A$ is \emph{atomic} if every 
\emph{read} and \emph{write} operation in $\EX$ is \emph{complete}
and there is a partial ordering $\prec_{\hist{\EX}}$ on all operations 
$\Pi$ in $\hist{\EX}$ such that: 
{\bf A1.} For any pair of operations $\op_1,\op_2\in\Pi$, 
if $\op_1\bef_{\hist{\EX}}\op_2$ then it cannot hold that $\op_2\prec_{\hist{\EX}}\op_1$, 
{\bf A2.} If $\op\in\Pi$ is a \emph{write} operation and $\op'$ any operation in $\Pi$,
then either $\op\prec_{\hist{\EX}}\op'$ or $\op'\prec_{\hist{\EX}}\op$, and 
{\bf A3.} If $v$ is the value returned by a \emph{read} $\rd$ then $v$ is the 
value written by the last preceding \emph{write} according to $\prec_{\hist{\EX}}$
(or the initial value $v_0$ if there is no such a write).}\vspace{-.2em}
\end{definition}

\remove{
In the case of a versioned R/W register not all write operations may modify the 
value of the register. Thus a \emph{write} (and the property {\bf A2}) refers to a 
$\trw{*}{*,chg}$ write operation that modifies the value (and the version) of the register. 
A \emph{read} (and {\bf A3}) refers to a $\act{cvr-read}$ or a $\trw{*}{*,unchg}$ operation
that does not modify the value (nor the version) of the register.
We say that, a write operation
 \emph{revises} a version $ver$ of 
the versioned register to a version $ver'$ (or \emph{produces} $ver'$) in an execution $\EX$, 
if $\trw{ver}{ver', chg}_{\pr_i}$ completes in $\hist{\EX}$.
%
}

\begin{figure}[!t]
	\begin{center}	
		\includegraphics[width=0.60\textwidth]{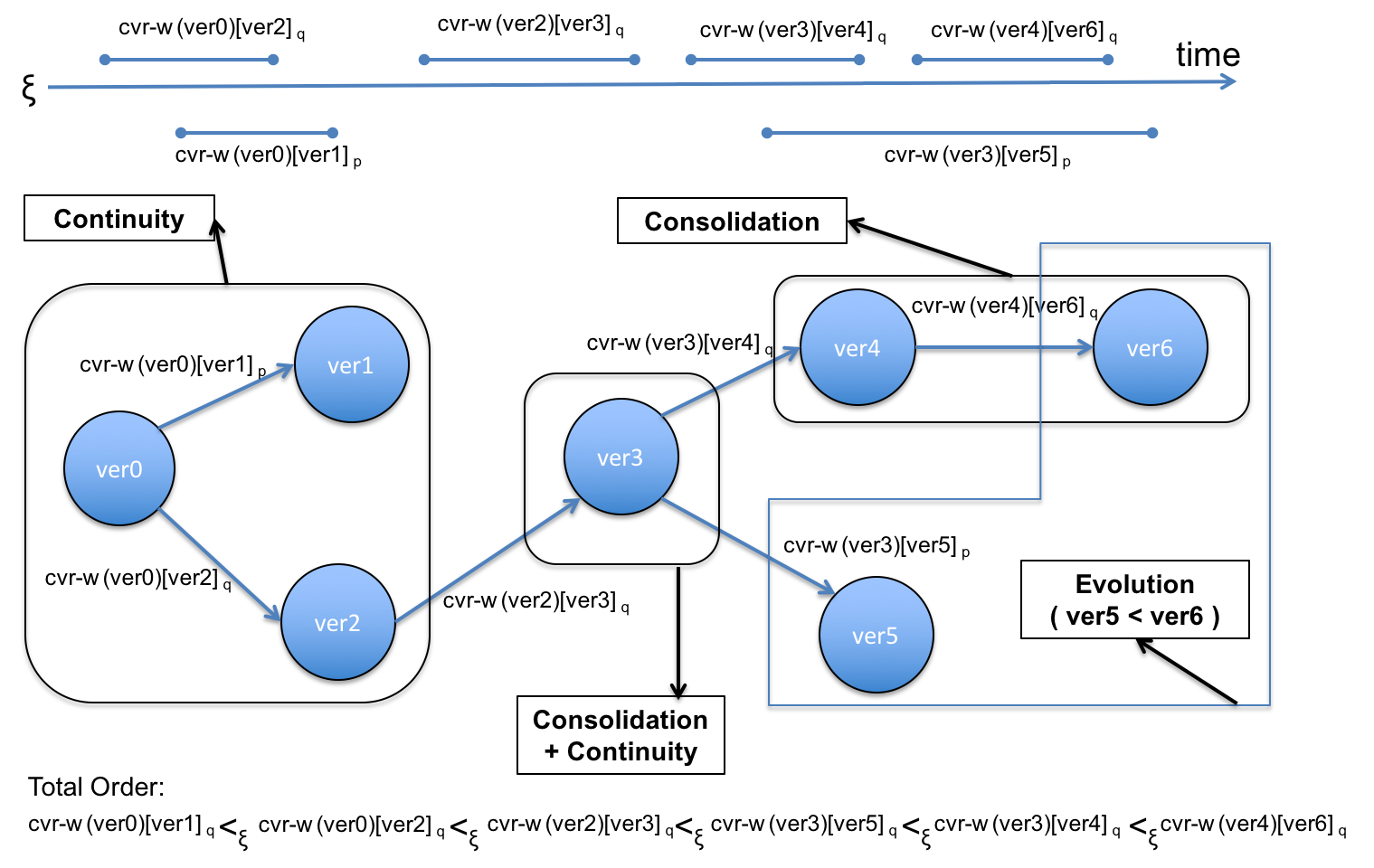}
	\end{center}\vspace{-1.5em}
	 \caption{Tree Illustration from Weak Coverable Execution
	 }
	 \label{fig:tree}\vspace{-1em}
\end{figure} 

In the context of versioned registers, in Definition \ref{def:atomic}, a \emph{write} refers to 
a successful write ($\trw{*}{*,chg}$) operation on the versioned register. Therefore, all 
the write operations in an execution $\EX$ are the ones that appear in $\wSet_{\EX,succ}$.
A \emph{read} refers to a versioned read ($\trrd{}{*}$) or an unsuccessful write ($\trw{*}{*,unchg}$) operation
that does not modify the value (nor the version) of the register.

\begin{definition}[Coverable atomic register]
 A versioned register is \textbf{(strongly/weakly) coverable} and \textbf{atomic}, referred as 
\emph{(strongly/weakly) coverable atomic register}, if any execution $\EX$ on the register satisfies:
(i) atomicity, and  (ii) strong/weak coverability (Definition \ref{def:strong}/\ref{def:weak}) with 
respect to the total order imposed by {\bf A2} on $\wSet_{\EX,succ}$. 
\end{definition}
Note that in a coverable atomic register, the ordering of read operations follows the ordering
from atomicity.
From this point onward, when clear from context, we refer to a coverable atomic register, as simply \emph{coverable register}. 
\vspace{-.2em}


\section{Weakly Coverable Atomic Registers vs Ranked Registers.}
\label{ssec:tr-vs-rr}
 
A type of registers that at first might resemble coverable registers are {\em ranked-registers}~\cite{GD05}.
As we show here, ranked-registers are weaker than {\em weakly} coverable registers. In particular, we show
that it is impossible to implement weakly coverable registers using ranked-registers; we begin by providing
a formal definition of ranked-registers.\vspace{-.2em}

\begin{definition}[Ranked-Registers~\cite{GD05}]
\label{def:rr}
Let $Ranks$ be 
a totally ordered set of ranks with $r_0$ the initial rank.
A ranked register is a MWMR shared object that offers the 
following operations: (i) $\act{rr-read}(r)$, with $r\in Ranks$ and returns $(r,v)\in Ranks\times Values$, and 
(ii) $\act{rr-write}(\tup{r,v})$, with 
$(r,v)\in Ranks\times Values$ and returns 
$commit$ or $abort$. A ranked register satisfies the 
following properties: 
	(i) {{\bf Safety.}} Every $\act{rr-read}$ operation returns a value and a rank that was written in some $\act{rr-write}$ invocation 
	or $(r_0, v_0)$. Additionally, if $W = \act{rr-write}(\tup{r_1, v})$ a write operation which commits 
	and $R = \act{rr-read}(r_2)$ such that $r_2 > r_1$, then $R$ returns $(r,v)$ where $r\geq r_1$.
	%
	(ii) {{\bf Non-Triviality.}} If a $\act{rr-write}$ operation $W$
	invoked with a rank $r_1$ aborts, then there exists an operation with rank $r_2 > r_1$ which returns before $W$ is invoked, or is concurrent with $W$ 
	%
	(iii) {{\bf Liveness.}} if an operation is invoked by a correct process then eventually it returns. \vspace{-.3em}

\end{definition}



We want to use rank-registers to implement the operations of
a weakly coverable register. As in Section \ref{sec:model}, 
we denote by $\trw{ver}{ver', flag}$ the 
coverable write operation that tries to revise version $ver$, 
and returns version $ver'$ with 
a $flag\in\{chg, unchg\}$. Similarly we denote by 
$\rrw{r}{r_h, res}$ a write operation on a ranked-register that 
uses rank $r$ and tries to modify the value of the register. The rank 
$r_h$ is the highest rank observed by an operation and $res\in\{abort, commit\}$.
%
In the following results we assume that a weakly coverable register is 
implemented using a set of ranked-registers. 
We begin with a lemma that shows that a coverable write operation 
revises the coverable register only if it invokes a write operation 
on some rank register and that write operation commits. {\em Omitted proofs
can be found in Appendix~\ref{appx:rr}.}

\begin{lemma}
\label{lem:commit}
Suppose there exists an algorithm $A$ that implements a weakly coverable register
using ranked-registers. In any execution $\EX$ of $A$, if
a process $\pr_i$ invokes a coverable write operation $\trw{ver}{ver',chg}_{\pr_i}$, 
then $\pr_i$ performs a write $\rrw{r}{r_h, commit}_{\pr_i,j}$ 
on some shared ranked-register $j$.
\end{lemma}

%

Next we show that if $\op_1, \op_2$ are two non-concurrent write operations  
on the weakly coverable register, then $\op_2$ 
performs a ranked write (that commits or aborts) on at least a single 
ranked register on which $\op_1$ performed a committed ranked write operation.
For the sake of the lemma $R_i$ is the set of ranked registers on which 
$\op_i$ writes, and $cR_i$ a subset of them on which the write commits.

\begin{lemma}
\label{lem:rrcommon}
Let $\op_1 = \trw{ver}{ver_1,chg}_{\pr_i}$ and 
$\op_2 = \trw{ver_1}{ver_2,*}_{\pr_z}, i\neq z$, be
two write operations that appear in an execution $\EX$ s.t. 
$\op_1\bef_{\EX}\op_2$.
There exists some shared register $j\in R_2\cap cR_1$ with 
a highest rank $r_j$ before the invocation of $\op_1$, such that 
$\pr_i$ performs an $\rrw{r}{*,commit}_{\pr_i,j}$ during $\op_1$, 
and $\pr_z$ performs an $\rrw{r'}{*,*}_{\pr_z,j}$ during $\op_2$.
\end{lemma}

Thus far we showed that a successful coverable write operation
needs to commit on at least a single ranked register (Lemma \ref{lem:commit}),
and two non-concurrent coverable write operations need to invoke a 
ranked write operation on a common rank register (Lemma \ref{lem:rrcommon}).   
Using now Lemma 
\ref{lem:rrcommon}
we can show that a coverable write operation that changes the version of 
the coverable register must use a rank higher than any previously
successful coverable write operation.

\begin{lemma}
\label{lem:rrorder}
In any execution $\EX$ if $\op_1 = \trw{ver}{ver_1,chg}_{\pr_i}$ and 
$\op_2 = \trw{ver_1}{ver_2,chg}_{\pr_z}$, $\!z\neq\! i$, s.t. 
$\op_1\bef_{\EX}\op_2$, then there exists some shared register $j$ such that 
$\pr_i$ performs an $\rrw{r}{*,commit}_{\pr_i,j}$ during $\op_1$, 
and $\pr_z$ performs an $\rrw{r'}{*,commit}_{\pr_z,j}$ during $\op_2$, and $r' > r$.
\end{lemma}

\remove{
\begin{lemma}
\label{lem:rrorder}
In any execution $\EX$ if $\op_1 = \trw{ver}{ver_1,chg}_{\pr_i}$ and 
$\op_2 = \trw{ver_1}{ver_2,chg}_{\pr_z}$, $\!z\neq\! i$, s.t. 
$\op_1\bef_{\EX}\op_2$, then there exists some shared register $j$ such that 
$\pr_i$ performs an $\rrw{r}{*,commit}_{\pr_i,j}$ during $\op_1$, 
and $\pr_z$ performs an $\rrw{r'}{*,commit}_{\pr_z,j}$ during $\op_2$, and $r' > r$.
\end{lemma}
}

%

Now we prove our main result stating that a weakly coverable 
register cannot be implemented with ranked registers as those
were defined in~\cite{GD05}.
 

\begin{theorem}
There is no algorithm that implements a weakly coverable register 
using a set of ranked registers. 
\end{theorem} 

\begin{proof}
The theorem follows from Lemmas \ref{lem:commit}, \ref{lem:rrcommon}, and \ref{lem:rrorder}, and the fact 
that a ranked register allows a write operation to commit even if it uses a rank smaller than the highest
rank of the register. As by Lemma \ref{lem:commit} a successful write must commit, then by ranked 
registers it can commit with a rank smaller than the highest rank of the accessed register. This, 
however, by Lemma \ref{lem:rrorder} may lead to violation of the consolidation and continuity 
properties and thus violation of weak coverability. 
\end{proof}
 
Observe that the key fact that makes ranked registers weaker than weakly coverable registers is that the former
allow write operations to commit even if their ranks are out of order. 
In particular, note that the Non-Triviality property \emph{does not force} a write operation 
invoked with a rank $r_1$ to abort, even if there exists a completed prior 
operation with rank $r_2 > r_1$. As shown in \cite{GD05} \emph{non-fault-tolerant} 
ranked registers may preserve the total order of the ranks, and thus be used to 
implement consensus. As we show in Appendix~\ref{appx:consensus} 
such ranked registers (i.e., that implement consensus) could be used to implement strongly coverable registers.

\section{Applications of Weakly Coverable Atomic Read/Write Registers}
\label{sec:applications}

\paragraph{Weak RMW registers.}

A shared object satisfies atomic \emph{read-modify-write} (RMW) semantics if 
a process can atomically \emph{read} and \emph{modify} 
the value of the object using some 
function $\mathcal{F}$, and then \emph{write} the new value on 
the object. 
%
Weakly coverable atomic R/W registers 
can be used to implement a weak version of RMW semantics. 
In a weak RMW object not all operations may successfully 
modify the value of the object. In case that a RMW operation
is not concurrent with any other operation then this operation
satisfies the RMW semantics. In case where two or more operations
invoke RMW concurrently, at least one of them will satisfy the 
RMW semantics. Finally, weak RMW allow multiple 
RMW operations to modify successfully the same value.  

Figure~\ref{fig:rmw} presents an implementation of a weak RMW object 
using weakly coverable atomic R/W registers. We assume that the object 
offers a $\act{rmw}(\mathcal{F})$ action that accepts a function 
and tries to apply that function on the value of the object.  
The object returns the initial value of the object and a flag 
indicating whether the value of the object was modified successfully.

\begin{figure}[!h]
	
	\hrule\vspace{0.15cm}
	\begin{footnotesize}
	At each process $i\in\idSet$\\
	Local Variables: $lcver\in Versions, oldval, lcval, newv \in Values, flag\in\{chg, unchg\}$\\
	
	{\bf function} {\sc Rmw}($\mathcal{F}$)
	\begin{algorithmic}
		\State $\tup{oldval, lcver}\gets \act{cvr-read}()$ 
		\State $newv \gets \mathcal{F}(oldval)$
		\State $\tup{lcval,lcver, flag} \gets \act{cvr-write}(lcver, newv)$
		\If { $flag == chg$ } return $\tup{lcval, success}$		
		\Else {} return $\tup{lcval, fail}$	
		\EndIf
	\end{algorithmic}
	\end{footnotesize}
	\hrule\vspace{.4em}
	\caption{Weak RMW using Weakly Coverable Atomic R/W Registers}\vspace{-1.2em}
	\label{fig:rmw}
\end{figure}

\begin{theorem}
\label{thm:wrmw}
The construction in Figure \ref{fig:rmw} implements 
a weak RMW object.
\end{theorem}

\begin{proof}
Consider an execution $\EX$ of the algorithm. 
We begin the proof by studying the case where an operation
$\act{rmw}(\mathcal{F})$ is not concurrent with any other operation
in $\EX$.
The atomic nature of the register ensures that $\act{cvr-read}$
returns the latest value and version, say $\tup{ver,val}$, written on the register. 
When the $\act{cvr-write}$ operation is invoked, the write operation
tries to modify the value associated with version $ver$. As there is 
no concurrent operation, the version of the register remains $ver$ 
and thus according to \emph{consolidation and continuity}, 
the write operation successfully writes the new value completing the 
RMW operation.

Consider now the case of two operations, $\op_1$ and $\op_2$, 
invoking $\act{rmw}$ concurrently. 
Each of these operations involve a $\act{cvr-read}$ 
followed by a $\act{cvr-write}$ operation. Let $\rd_{\op_i}$ (resp. $\wrt_{\op_i}$)
denote the read (resp. write) operation invoked during $\op_i$, for $i\in[1,2]$. We have the 
following cases wrt the order of these operations:
$(i)$ $\wrt_{\op_1} \bef \rd_{\op_2}$,
$(ii)$ $\wrt_{\op_2} \bef \rd_{\op_1}$,
$(iii)$ $\rd_{\op_2}\bef\wrt_{\op_1} \bef \wrt_{\op_2}$,
$(iv)$ $\rd_{\op_1}\bef\wrt_{\op_2} \bef \wrt_{\op_1}$, or
$(v)$ $\wrt_{\op_1}$ is concurrent with $\wrt_{\op_2}$.
In case $(i)$, both read and write operations of
$\op_1$ complete before the read and write operations
of $\op_2$ are invoked. In this case notice that 
the version of the object remains the same from the read
to the write operation of both operations. Thus, 
according to \emph{consolidation and continuity}, both write operations
will successfully change the value of the register. 
The same holds for case $(ii)$, where $\op_2$'s ops 
complete before the invocation of $\op_1$'s ops. 
In case $(iii)$ the write operation of $\op_1$ completes 
before the write operation of $\op_2$. Let $\rd_{\op_2}$
in this case complete before $\wrt_{\op_1}$. Both read 
operations $\rd_{\op_1}$ and $\rd_{\op_2}$ discover by \emph{atomicity} 
the same version, say $ver$. So both write operations will 
be invoked as $\act{cvr-write}(ver, v)$. Since no operation 
changes the version of the register before $\wrt_{\op_1}$ is 
invoked, then by \emph{consolidation and continuity}, $\wrt_{\op_1}$ changes the version 
of the object to, say, $ver_{\op_1}$. Notice that 
by \emph{validity}, $ver_{\op_1}>ver$. When $\wrt_{\op_2}$ is
invoked it fails by \emph{consolidation} to change the value of
the object as $\wrt_{\op_1} \bef \wrt_{\op_2}$ and it tries 
to change the version $ver < ver_{\op_1}$ (the version of $\wrt_{\op_1}$). Hence,
only $\op_1$ will manage to preserve RMW semantics. Similarly,
we can show that only $\op_2$ will preserve RMW semantics in 
case $(iv)$. 
Finally, in case $(v)$ if both writes try to change the version 
$ver$,  both may succeed and preserve 
RMW semantics. Since, however, their versions are unique and 
comparable, then by \emph{consolidation} any subsequent operation will RMW the 
highest of the two versions. 
So in all cases at least a single 
operation satisfies the RMW semantics, as desired.
\end{proof}

From the proof we can extract that weakly coverable registers 
may allow multiple writes to change the same version
of the register, but
\emph{consolidation} ensures that at least one write satisfies 
RMW semantics for each version. 
Finally,  \emph{consolidation and continuity} ensure that eventually RMW operations 
diverge in a single path in the constructed tree.
\vspace{-.5em}

\paragraph{Concurrent File Objects}

A file object can be implemented directly using RMW semantics since one can retrieve, revise, and write back the new version of the file. As RMW semantics can be used to solve consensus \cite{H91}, 
they are impossible to be implemented in an asynchronous
system with a single crash failure. 
Therefore, we consider file objects that comply to the weak
RMW semantics as those were given in the paragraph above. 
In particular,  
we consider \emph{concurrent file objects} that allow two fundamental operations, \emph{revise} and 
\emph{get} to be invoked concurrently by multiple processes. The $\act{revise}$ operation is used to change the contents of the file object,
whereas the $\act{get}$ action is analogous to a read operation and facilitates the retrieval
of the contents of the file. Semantically, a file object requires that a revise operation is 
applied on the latest version of the file and a get operation returns the file associated
with the latest written version.  
Depending on the implementation, the values written and returned
by these operations can be the complete file object, a fragment of the file object, or just 
the journal containing the operations to be applied on a file (similar to a journaled file system).

Figure \ref{fig:file} presets the algorithm that implements the two operations. The \emph{revise} 
operation specifies the version of the file to be revised along with the new 
value of the shared object. 
The $\act{cvr-write}$ operation attempts to perform the write with the given  version and 
returns the value and version of the register, and whether the write succeeded or not. 
If the write succeeded then the operation informs the application 
for the proper completion of the revise operation; otherwise the latest discovered
value-version pair is returned. From Theorem~\ref{thm:wrmw} and Figure~\ref{fig:file}
we may conclude the following theorem.\vspace{-.3em}

\begin{figure}[t]
	
	\hrule\vspace{0.15cm}
	\begin{footnotesize}
	\begin{multicols}{2}
	At each process $i\in\idSet$\\
	Local Variables: 
	$lcver \in Versions$, initially $ver_0$
	$lcval, newv \in Values$, initially $\bot$
	$flag\in\{chg, unchg\}$, initially $chg$\\
	
	{\bf function} {\sc Revise}($v,ver$)
	\begin{algorithmic}
			\State $\tup{lcval,lcver, flag} \gets \act{cvr-write}(ver, v)$
			\If {$flag == chg$}
				\State return OK
			\EndIf		
		\State return $\tup{lcval, lcver}$\\
	\end{algorithmic}
	
	{\bf function} {\sc Get}()
	\begin{algorithmic}
		\State $\tup{lcval, lcver}\gets \act{cvr-read}()$ 
		\State return $\tup{lcval, lcver}$
	\end{algorithmic}
	\end{multicols}
	\end{footnotesize}
	\hrule\vspace{.4em}
	\caption{File Object using Weakly Coverable Atomic R/W Registers}\vspace{-2em}
	\label{fig:file}
\end{figure}

\begin{theorem}
\cg{The construction in Figure \ref{fig:file} implements 
a file object.}
\end{theorem}


\section{Implementing Weakly Coverable Atomic Read/Write Registers\vspace{-.5em}}
\label{sec:algorithms}

\cg{We now show how we can implement weak coverable atomic registers. We do so by 
enhancing the Multi-Writer version of algorithm ABD~\cite{ABD96,LS97} (referred as \mwABD)
to preserve the properties of weak coverability.} 
The presented technique can be applied to implementations
of atomic R/W objects that utilize a $\tup{tag,value}$ pair 
to order the 
write operations and where each write performs 
two phases before completing: a \emph{query phase} to obtain 
the latest value of the atomic object and a \emph{propagation 
phase} to write the new value on the object. 
We could also adopt implementations of 
stronger objects like the ones presented  
in \cite{BDFG03, GD05, CDS13-BA, DVV14} but we 
preferred to show the simplest modification in
a fundamental algorithm. 
%
\remove{
Our algorithm aims for the
strategy ``hope for the best" and allows write operations to propagate values that can be hidden
to succeeding \emph{reads} or \emph{writes}. Our technique, however, always 
provides \emph{provable} guarantees on the state of the object as it 
does not compromise atomicity.
In the next section we provide an enhancement on the basic ABD algorithm
\cite{ABD96} and in the following section we show how our technique can
be used to enhance atomic implementations designed for large objects
\cite{FL03}. We show how coverability helps implementations of large 
data objects to support the management and maintenance of \emph{file objects}.
} 
%
%
To capture the semantics of a coverable atomic register we modify the operations 
of algorithm \mwABD{} 
to comply with the versioned variant of the R/W register. 
We use $\act{cvr-write}(ver, v)$ and $\act{cvr-read}()$ as the
write and read operations respectively. 
A  $\act{cvr-write}(ver, v)$ operation \cg{may impact differently the state of the object,}
depending on the version of the shared object: 
it may appear as a \emph{read} not modifying the value nor the version of the 
register or as a \emph{write} changing both the value and the 
version of the register. 

In brief, the original \mwABD{} replicates an object to a set of hosts 
$\srvSet\subset \idSet$ and it uses $\tup{tag, value}$ pairs to order the 
\emph{read} and \emph{write} operations. 
A $tag$ consists of a \emph{non-negative integer} number and a \emph{writer identifier} which is used to break the 
ties among concurrent write operations. 
Both the read and write protocols have two phases: a \emph{query} and a \emph{propagation} phase.  
During the \emph{query} phase the invoking process 
broadcasts a query message to all the replica hosts (replicas)  and 
waits for a majority of them to reply with their 
tag-value pairs. Once those replies are received the 
process discovers the largest tag-value pair among the replies.
In the second phase, a read operation 
propagates the discovered tag-value pair to the majority of the replicas. 
A write operation increments the largest tag, associates the 
new tag with the value to be written, and propagates the new 
tag-value pair to the majority of the replicas.

In the \emph{versioned \mwABD}, \cg{v\mwABD\ for short}, we use the tags associated with 
each value to denote the version of the register. The pseudocode 
of each operation of v\mwABD\ is described in Figure~\ref{fig:abd}.
The $\act{cvr-read}$ operation is similar to the read
operation of \mwABD{} with the difference that it 
returns both the value and the version of the register.  
A $\act{cvr-write}$ operation differs from the original write
by utilizing a condition before its \emph{propagation} phase
\cg{and depending whether the condition holds it changes the state
of the register (value and version) or not,  
as detailed in Figure~\ref{fig:abd}.} Note that 
the version parameter of the write operation is equal to
the maximum tag that the invoking process witnessed.\vspace{-.3em}
\remove{
In particular,
a $\act{cvr-write}$ is executed in two phases. During the
first phase it collects the value of the object from a majority of replicas
and discovers the maximum tag among the replies. Before 
moving to a second phase the $\act{cvr-write}$ checks whether 
its local tag is equal to the maximum tag discovered in
the first phase. If this is the case then the $\act{cvr-write}$
acts as the original \mwABD{} write operation, 
by incrementing the tag and sending the new tag along with the value to be written to a majority of 
replicas. If the condition does not hold then the $\act{cvr-write}$ acts as
the second phase of a read operation in the original \mwABD{} and propagates the 
maximum tag together with its associated value to the majority of replicas. 
In the first case the write operation returns the value and version written 
along with the flag $chg$ to indicate the modification of the register. 
In the latter case the operation returns the maximum tag-value pair discovered 
in its query-phase along with the flag $unchg$ to indicate that the 
register was not modified. 
}

\begin{figure}[t]
\hrule\vspace{0.15cm}
\begin{footnotesize}

			
		$\act{cvr-write}(val, ver = maxtag)$
			\begin{itemize}[leftmargin=5mm]
				\item[] \act{query-phase}: Send query request to \emph{all} the replicas and wait to 
				receive $(tag,value)$ responses from a majority of them. Select the 
				$(tag,value)$ among the collected replies with the largest tag; 
				let the $\tup{\tau,v}$ be this pair and the integer component of $\tau$ be $z$. Then: \vspace{-0.5em}
				
				\begin{itemize}[leftmargin=5mm]
					\item If $ver == \tau$	then: 
					Create a new tag $\tau_{new} = \tup{z+1, wid}$ 
					where $wid$ is the unique identifier of the writer
					and set $val_{new} = val$. 
					
					\item If $ver \neq \tau$ then: Set $\tup{\tau_{new}, val_{new}} = \tup{\tau,v}$. \vspace{-0.5em}
					
				\end{itemize}
				
				\item[] \act{propagation-phase}: 
					Send $\tup{\tau_{new}, val_{new}}$ to all the replicas and wait to 
					receive responses from a majority of them. 
					if $ver == \tau$ then respond with $\tup{val_{new}, \tau_{new}, chg}$,
					otherwise respond with $\tup{val_{new}, \tau_{new}, unchg}$ to the process. 
			\end{itemize}
			
			$\act{cvr-read}()$
			\begin{itemize} [leftmargin=5mm]
				\item[] \act{query-phase}: Send query request to \emph{all} the replicas and wait to 
				receive $(tag,value)$ responses from a majority of them. Select the 
				$(tag,value)$ among the collected replies with the largest tag; 
				let the $\tup{\tau,v}$ be this pair and the integer component of $\tau$ be $z$.
				
				\item[] \act{propagation-phase}: 
					Send $\tup{\tau, v}$ to all replicas and wait for responses from a majority of them. 
					Respond with $\tup{v, \tau}$ to the process. 
				\end{itemize}
				
			\act{at-replica}
			\begin{itemize} [leftmargin=5mm]
				\item[] On receipt of \act{query} message: Send the tag-value 
				pair $\tup{\tau_{r},v_r}$ stored locally.
				
				\item[] On receipt of \act{propagation} message: Let $\tup{\tau_m,v_m}$ be the tag-value
				pair enclosed in the received message and $\tup{\tau_r,v_r}$ the local pair on the replica. 
				Compare the tags $\tau_m$ and $\tau_{r}$. If $\tau_{m} > \tau_{r}$ then 
				store $\tup{\tau_m,v_m}$ locally. Reply with ``ack".
			\end{itemize}
		\end{footnotesize}
		\hrule\vspace{.2em}
\caption{The operations of algorithm v\mwABD.}
\label{fig:abd}\vspace{-1em}
\end{figure}

\begin{theorem}
Algorithm \rm{v}\mwABD{} implements weak coverable atomic registers.
\end{theorem} 

\begin{proof}
It is clear that \rm{v}\mwABD{} still satisfies properties {\bf A1-A3}. Any write operation 
that is not successful can be mapped to a read operation that performs two phases and propagates the latest
value/version of the register to a majority of replicas 
before completing.  
It remains to show that 
\rm{v}\mwABD{} also satisfies 
the properties of validity and weak coverability. 

\emph{Validity} is satisfied since each tag is unique, as it is composed by 
an integer and the id of a process. The tag is monotonically incrementing
at each replica, as according to the algorithm a replica updates its local 
copy only if a higher tag is received. A writer process discovers the maximum 
tag $maxtag$ among the replicas and in the second phase it generates a tag $\tup{maxtag+1,wid}$. 
As the tag at each replica is monotonically incrementing then each writer never 
generates the same tag twice. Also, for every write 
$\trw{\tg{}}{\tg{}',chg}$, $\tg{}'=\tup{\tg{}.ts+1, wid} \Rightarrow \tg{}'>\tg{}$. 
Finally, since every tag is generated by extending the initial tag and each 
write operation extends a tag that obtains during its query phase then 
there is a sequence of tags leading from the initial tag to the tag used
by the write operation.

For \emph{consolidation} we need to show that for two write 
operations $\wrt_1=\trw{*}{\tg{1},chg}$ and $\wrt_2=\trw{\tg{2}}{*,chg}$,
if $\wrt_1\bef_{\EX}\wrt_2$ then $\tg{1} \leq \tg{2}$.  
According to the algorithm $\wrt_1$ propagates $\tg{1}$ to the 
majority of replicas before completing. 
In the query phase, $\wrt_2$ receives messages from the 
majority of replicas. So there is one replica $s$ that received
$\tg{1}$ from $\wrt_1$ before replying to $\wrt_2$. Since 
the $\tg{}$ in $s$ is monotonically incrementing, then $s$ 
replies to $\wrt_2$ with a tag $\tg{s}\geq\tg{1}$.
So $\wrt_2$ receives a $maxtag\geq \tg{1}$.
Since $\wrt_2$ also changes the value and version of the 
register it means that its local tag $\tg{2}$ is equal to $maxtag$. 
This shows immediately that $\tg{2} \geq \tg{1}$, completing the proof.

\emph{Continuity} is preserved as a write operation first queries the 
replicas for the latest tag before proceeding to the propagation phase 
to write a new value. Since the tags are generated and propagated only 
by write operations then if a write changes the value of the system then 
it appends a tag already written, or the initial tag of the register.

Finally, to show that \emph{evolution} is preserved, we observe that the version of a register is given by
its tag, where tags are compared lexicographically (first the number $\tg{}.ts$ and then the writer identifier to break ties).
A successful write $\op_1=\trw{\tg{}}{\tg{}'}$ generates a new tag $\tg{}'$ from $\tg{}$ such
that $\tg{}'.ts = \tg{}.ts + 1$. Consider sequences of tags $\tg{1}, \tg{2},\ldots, \tg{k}$
and $\tg{1}', \tg{2}',\ldots, \tg{\ell}'$ such that $\tg{1}=\tg{1}'$.
Assume that $\trw{\tg{i}}{\tg{i+1}}$, for $1\leq i<k$, and $\trw{\tg{i}'}{\tg{i+1}'}$, for $1\leq i<\ell$,
are successful write operations.
If $\tg{1}.ts=\tg{1}'.ts=z$, then $\tg{k}.ts=z+k$ and $\tg{\ell}'.ts=z+\ell$, and if $k < \ell$ then $\tg{k} < \tg{\ell}'$.\vspace{.3em}
\end{proof}

\noindent{\bf Supporting Large Versioned Objects.}
Fan and Lynch~\cite{FL03}, using algorithm \mwABD\ as a building block, 
showed how large atomic R/W objects can be efficiently replicated. The main idea 
of their algorithm, called LDR, is to have
two distinguished sets of servers: Replicas and Directories. Replica servers
are the ones that actually store the object's data (value), while Directories keep track
of the tags of the object and the associated Replicas that store the data of
the object. A reader or writer first runs algorithm \mwABD\ on the Directories to
obtain the highest tag of the object, and the identity of the Replicas that have the associated
value (aka, the most recent value of the object). A read operation, then
contacts a subset of the Replicas to obtain the value of the object. 
A write sends the new value to a majority of the Replicas, while ensuring
that Directories are updated (see~\cite{FL03} for details). 
By replacing algorithm \mwABD\ with algorithm v\mwABD\ and performing a few modifications
to the Replicas, we can turn algorithm LDR into an algorithm that can handle
{\em large versioned R/W objects}, such as large files. 
See Appendix~\ref{appx:files} for the modified LDR. \vspace{-.5em}

\section{Conclusion}
\label{sec:conclude}

In this paper we have introduced {\em versioned registers} and a new property 
for concurrent versioned registers, we call \emph{coverability}. 
A versioned register associates a version with its value, and with each operation that wants to modify its value.
An operation may modify the value and the version of the register, or it may just retrieve its value-version pair.  
Coverability defines the exact guarantees that a versioned register 
provides when it is accessed concurrently by multiple processes with respect to the evolution of its versions,
over a total order of its operations.
We introduce two levels of coverability: \emph{strong} and \emph{weak}. Strong coverability 
requires that only a single operation modifies each version of the register, whereas weak coverability
is more relaxed allowing multiple concurrent operations to modify the same version.

We combine coverability with atomicity to obtain (strongly/weakly) coverable atomic registers. 
The successful writes on the register follow the total order of atomicity, while preserving the
properties required by coverability. We note that a different total ordering could be used with
coverability to obtain other types of ``coverable objects". In fact, we believe it would be interesting 
to investigate further the use of coverable objects for the introduction of distributed algorithms 
for various applications. The fact that each operation is enhanced by the version of the object provides 
the flexibility to manipulate the effect of an operation under some conditions on 
the version of the object with respect to the version of the operation.

\newpage


\bibliographystyle{acm}
\bibliography{biblio}

\newpage
\appendix
\section*{Appendix}

\section{Impossibility of Implementing Weakly CoVerable 
Registers using Ranked Registers}
\label{appx:rr}
In this section we provide the proofs to the lemmas presented 
in Section \ref{ssec:tr-vs-rr}. Before proceeding to the proofs 
let us introduce some notation we use throughout this section. 

Let $R$ be a set of ranked registers. Let $R_x\subseteq R$ denote the set 
of ranked registers on which a process $\pr_i$ 
performs a $\rrw{r}{*,*}_{\pr_i,*}$ during a coverable write operation $\op_x$
in an execution $\EX$. $R_x=cR_x\cup aR_x$, where $cR_x$ is the set of ranked 
register such that $\pr_i$ performs a $\rrw{r}{*,commit}_{\pr_i,*}$ that commits 
during $\op_x$, and 
$aR_x$ the set of ranked register such that $\pr_i$ performs 
a $\rrw{r}{*,abort}_{\pr_i,*}$ that aborts during $\op_x$. 
For any pair of write operations $\op_x\bef_{\EX}\op_y$,
let the set $R_{x,y} = R_x \cap R_y$ be the set of ranked registers such that 
both $\op_x$ and $\op_y$ perform a ranked write. We finally denote by
$cR_{x,y} = cR_x \cap R_y$ and $aR_{x,y} = aR_x \cap R_y$, the 
set of registers where $\pr_i$ committed (or aborted resp.) during
$\op_x$ and they were also written during operation $\op_y$.\vspace{1em}


\noindent {\bf Proof of Lemma \ref{lem:commit}.}
Let the weakly coverable register be implemented by $k$ 
ranked registers each with a highest rank $r_1, r_2,\ldots,r_k$ 
respectively at the end of some execution fragment $\EX$. 
For the rest of the proof we will construct extensions of $\EX$.  
Also, let the state of the coverable object be $(v, ver)$ at the 
end of $\EX$.

Assume to derive contradiction that we extend $\EX$ 
with a write operation $\wrt_1 = \trw{ver}{ver',chg}_{\pr_i}$
that revises the coverable register, and 
all the write operations performed during $\wrt_1$ on the ranked 
registers abort. From that it follows that for each 
write operation $\rrw{r}{r_j,abort}_{\pr_i,j}$ performed 
by $\pr_i$ on some register $j$, $r_j > r$. 
Let the new execution be $\EX'$. 

We extend $\EX'$ with another write operation 
$\wrt_2 = \trw{ver''}{ver''',chg}_{\pr_z}$ by process $\pr_z$ to obtain 
execution $\EX''$. \nn{Since $\wrt_1\bef_\EX\wrt_2$ then by \emph{consolidation},
$ver''\geq ver'$, and $\wrt_1<_\EX\wrt_2$. Moreover, since $\wrt_1$ is not concurrent 
with any other operation, then 
by consolidation $ver'$ is the largest version introduced in $\EX$. 
Since, by \emph{continuity}, 
$ver''$ has to be equal to a version introduced by a preceding
operation, and since $ver''\geq ver'$ (the largest version introduced), then  
$\wrt_2$ revises $ver'' = ver'$ to a newer version $ver'''$.}
Note however that for any write operation $\rrw{r'}{r_j,*}_{\pr_z, j}$ 
performed on any of the ranked registers, for $1\leq j\leq k$,
the highest rank for $j$ at the time of the write was $r_j$. 

Finally consider the execution $\Delta\EX''$ that is similar to 
$\EX''$ without containing $\wrt_1$. In other words $\Delta\EX''$
extends $\EX$ with the write operation $\wrt_2$. Observe that any 
write operation $\rrw{r'}{r_j, *}_{\pr_z, j}$ performed by $\pr_z$ on 
ranked register $j$ during $\wrt_2$ observes a highest rank $r_j$ 
as in $\EX''$. So if such a write committed (or aborted) in $\EX''$ 
will also commit (or abort) in $\Delta\EX''$ as well. Therefore,
since $\wrt_2$ revised the value of the coverable register in $\EX''$ 
will revise the value of the coverable register in $\Delta\EX''$ 
as well. \nn{However, the last proceeding write operation is of the 
form $\trw{*}{ver,*}$ for $ver\neq ver'$. 
%
Thus $\Delta\EX''$ violates the \emph{continuity} property and hence contradicts our 
initial assumption.}\hfill$\square$\vspace{1em}

\noindent {\bf Proof of Lemma \ref{lem:rrcommon}.}
We will assume to derive contradiction that $\op_2$ does not write on any 
ranked register that $\pr_i$ wrote (and committed) during $\op_1$. 
More formally, let $cR_1$ be the set of ranked registers
s.t. for all $j\in cR_1$, $\rrw{r}{*,commit}_{\pr_i,j}$ for some rank 
$r$ during $\op_1$.  Let $R_2$ be the set of ranked registers 
s.t. for all $q\in R_2$, $\pr_z$ invokes $\rrw{r'}{*,*}_{\pr_z,q}$ 
during $\op_2$. Note that since $\op_1$ 
revises the version of the object, 
then according to Lemma \ref{lem:commit}, $|cR_1| \geq 1$.  
According to our assumption $cR_1\cap R_2 = \emptyset$. 

Let us now construct an execution that contains the two operations $\op_1$ and $\op_2$.
Consider an execution fragment $\EX$ that ends with a state associated with a version $ver$. 
Let us assume that there exists an algorithm $A$ that uses $k$ ranked registers 
each with a highest rank $r_1, r_2,\ldots,r_k$ respectively at the end of $\EX$. 
We extend $\EX$ with operation $\op_1$ and obtain $\EX_1$. Since $\op_1$ 
changes the version of the object, by Lemma \ref{lem:commit},
there exists a ranked register $j\in cR_1$ such that, $\pr_i$ invokes an operation that commits on $j$, 
$\rrw{r}{*,commit}_{\pr_i,j}$, during $\op_1$. 

Next we extend $\EX_1$ by $\op_2$ and obtain $\EX_2$. 
Since according to our assumption, $cR_1\cap R_2=\emptyset$, then it 
must be the case that the highest rank observed by $\op_2$ in any
$j\in R_2$ is $r_j$, i.e. the highest rank of $j$ at the end of $\EX$.
So it returns either $r_j$ or $r'$ the rank used by $\pr_z$.
That includes also the ranked registers that $\pr_i$  tried to modify and 
aborted during $\op_1$. 


Consider now the execution $\Delta\EX_2$ which is similar to $\EX_2$,
without operation $\op_1$. In particular, $\Delta\EX_2$ is obtained 
by extending $\EX$ with $\op_2$. Notice that since $\pr_z$
does not communicate with $\pr_i$, then $\pr_z$ appears in the 
same state in both $\EX_2$ and $\Delta\EX_2$ before invoking $\op_2$. 
Thus, $\pr_z$ attempts to write on  
the same set of ranked registers $R_2$ in both executions. 
Since $\EX$ is extended by $\op_2$ alone, then any write operation 
on the ranked registers $j\in R_2$ is $r_j$ (as in $\EX_2$). 
So $\op_2$ cannot distinguish $\EX_2$ 
from $\Delta\EX_2$ and thus revises $ver_1$ in $\Delta\EX_2$ as well. 
However, $\Delta\EX_2$ does not contain a $\trw{*}{ver_1,chg}$ operation, therefore $\op_2$ violates the \emph{continuity}
property of weak coverability. 
This contradicts our assumption. 
\hfill$\square$\vspace{1em}

\noindent {\bf Proof of Lemma \ref{lem:rrorder}.}
	Consider again an execution fragment $\EX$ that ends with a state associated with a version $ver$. 
	Let us assume that there exists an algorithm $A$ that uses a set $|R|=k$ of ranked registers 
	each with a highest rank $r_1, r_2,\ldots,r_k$ respectively at the end of $\EX$. 
	We know by Lemma \ref{lem:commit}, that each operation $\op_x$ that changes 
	the version of the weakly coverable register performs a write 
	that commits on at least a single ranked register in $R$. 
	
	We extend $\EX$ with the following non-concurrent operations (listed in the order they take place) 
	to obtain execution $\EX_1$:
	\begin{itemize}
		\item operation $\op_1 = \trop{ver}{ver_1,chg}_{\pr_1}$
		\item operation $\op_2 = \trop{ver_1}{ver_2,chg}_{\pr_2}$
	\end{itemize}

	By Lemma \ref{lem:rrcommon}, $cR_{1}\cap R_2\neq\emptyset$.
	Assume to derive contradiction that $\forall j \in cR_1\cap R_2$, $\pr_1$
	performs a committed write with a rank $r_{\op_1}>r_j$ and $\pr_2$ performs a write 
	with a rank $r_{\op_2} < r_{\op_1}$  (that may commit or not). 
	Since according to our assumption, 
	$\forall j \in cR_1\cap R_2$, the rank of $\pr_2$ has to be smaller than 
	the rank used by $\pr_1$, we assume w.l.o.g. that $\pr_2$ uses the same 
	rank $r_2$ for all the ranked writes. 

	By the order of operations in $\EX_1$ it follows that for all $j\in cR_{1}\cap R_2$, 
	$\rrw{r_{\op_1}}{r_j,commit}_{\pr_1,j}$ appears before $\rrw{r_{\op_2}}{r_1,*}_{\pr_2,j}$
	in $\EX_1$. Moreover, observe that, by Definition \ref{def:rr}, for each register $i\in R_1\setminus cR_1$,
	$r_i>r_{\op_1}$ since the write from $\op_1$ aborted. 
	Since $\op_2$ changes 
	the version of the weakly coverable register, then by Lemma \ref{lem:commit}, $cR_2\neq\emptyset$. 
	Notice that, even though we assume that $r_2 <  r_1$, 
	the operations in $\EX_1$ may commit
	without violating the ranked register properties of Definition \ref{def:rr} (as a write operation
	with a smaller rank does not have to abort). In order to preserve 
	weak coverability, $\op_2$ changes the version $ver_1$ to $ver_2$. 
	
	Consider now the execution $\Delta\EX_1$ that contains the same 
	operations but with $\op_1$ and $\op_2$ in reverse order. 
	In particular $\Delta\EX_1$ extends $\EX$ with operations:
	\begin{itemize}
		\item operation $\op_2 = \trop{ver_1'}{ver_2',chg}_{\pr_2}$
		\item operation $\op_1 = \trop{ver}{ver_1,chg}_{\pr_1}$
	\end{itemize}  
Since there is no communication assumed between the processes then $\op_2$ uses 
rank $r_{\op_2}$ in $\Delta\EX_1$ as well. It is easy to see that for any register $i\in R_2\setminus R_1$,
$\op_2$ observes the same highest rank $r_i$ in both executions $\EX_1$ and $\Delta\EX_1$.
So if the rank write of $\op_2$ on those registers commits in $\EX_1$ then it also commits in $\Delta\EX_1$.
So the only registers that may allow $\op_2$ to differentiate between the two executions are the 
ones in the intersection $cR_1\cap R_2$. There are two cases to consider: (i) $\forall j\in cR_1\cap R_2$,
$r_j > r_{\op_2}$, and (ii) $\exists j\in cR_1\cap R_2$, and $r_j\leq r_{\op_2}$. 

\case{(i)} In case (i), $\op_2$ witnesses
a higher rank from all the registers in $cR_1\cap R_2$ as in $\EX_1$. So for each register $j\in cR_1\cap R_2$,
if $\rrw{r_{\op_2}}{r_j,*}_{\pr_2,j}$ committed in $\EX_1$ then the write commits in $\Delta\EX_1$ as well.
Thus, $\op_2$ will not be able to distinguish the two operations and it extends $ver_1' = ver_1$ in 
$\Delta\EX_1$ as well. However, $ver_1$ is not written in $\Delta\EX_1$, thus $\op_2$ violates \emph{continuity}
property and contradicts our assumption. 

\case{(ii)} So it remains to examine the second case were $\exists j\in cR_1\cap R_2$, and $r_j\leq r_{\op_2}$.
In this case $\rrw{r_{\op_2}}{r_j,*}_{\pr_2,j}$ has to commit in $\Delta\EX_1$. If the same operation committed 
in $\EX_1$ as well then $\op_2$ cannot distinguish the two executions and thus violates coverability as shown before. 
Let us assume that $\op_2$ did not commit in $\EX_1$. Hence, $\op_2$ distinguishes $\Delta\EX_1$ from $\EX_1$.
To preserve weak coverability, $\op_2$ has to extend version $ver_1'= ver$ to a version $ver_2' > ver$. At the end 
of $\op_2$ the highest rank of $r_j = r_{\op_2}$. When $\op_1$ is invoked it performs rank writes using rank 
$r_{\op_1}$, since there is no communication between the processes. Since, according to our assumption 
$r_{\op_2} < r_{\op_1}$, it follows that $\rrw{r_{\op_1}}{*,*}_{\pr_1,j}$ commits in 
both $\EX_1$ and $\Delta\EX_1$. Moreover, since for all the rest registers $i\in cR_1$, $r_i > r_{\op_2}$,  
$\op_1$ will witness the same highest rank $r_i$ from each of those registers, in both executions. Thus, 
all the write operations on those registers $\op_1$ will commit on all those registers, and thus, $\op_1$ 
will not be able to distinguish $\Delta\EX_1$ from $\EX_1$. Since, however ,it extended $ver$ in $\EX_1$,
then it extends $ver$ in $\Delta\EX_1$ as well. However, as $\op_2\bef_\EX \op_1$,  
then by consolidation, $\op$ needs to extend a version larger or equal to $ver_1'$. Since $ver < ver_1'$ 
then consolidation is violated.  And this completes the proof.
\hfill$\square$ 

\section{Strong Coverability vs Consensus}
\label{appx:consensus}
Consensus \cite{Lynch1996} is defined as the problem where a set of fail-prone processes
try to agree on a single value for an object. A consensus protocol
must specify two operations: (i) $propose(v)_{\pr_i}$, used by the process $\pr_i$ to propose
a value $v$ for the object, and (ii) $decide()_{\pr_i}$, used by the process $\pr_i$ to decide 
the value of the object. Any implementation 
of consensus must satisfy the following three properties: 
\begin{itemize}[leftmargin=5mm]\itemsep2pt
	\item 
{\bf (1) CTermination:} Every correct process decides a value;
	\item 
	{\bf (2) CValidity:} Every correct process decides at most 
	one value, and if it decides some value $v$, then $v$ must have
	been proposed by some process;  
	\item 
	{\bf (3) CAgreement:} All correct process must decide the same value.
\end{itemize}

We show that a strongly coverable atomic register 
 is equivalent to a consensus object. To support this statement
we first present an implementation of a consensus object using 
a strongly coverable register, and then we describe an
implementation of a strongly coverable 
register assuming the existence of a consensus object. 
In the implementation of consensus that follows we assume that all the 
processes propose a value and they decide by the end of the propose operation. 
Thus we combine the two actions in one operation.
Figure \ref{fig:consensus}
presents the pseudocode of the implementation of a consensus object 
using a strongly coverable atomic register.  

\begin{figure}[!ht]
	
	\hrule\vspace{0.15cm}
	\begin{footnotesize}
	At each process $i\in\idSet$\\
	Local Variables: $lcver\in Versions, lcval\in Values, flag\in\{chg, unchg\}$\\
	
	{\bf function} {\sc Propose}($v$)
	\begin{algorithmic}
		\State $lcval\gets v$ 
		\State $(lcval, lcver, flag) \gets \act{cvr-write}(lcval,ver_0)$
		\State return $lcval$
	\end{algorithmic}

\end{footnotesize}
	\hrule\vspace{.4em}
	\caption{Consensus using Strongly Coverable Atomic Registers}\vspace{-1em}
	\label{fig:consensus}
\end{figure}

We assume that $ver_0$ is the initial 
version of the coverable register. When each process begins
executing the algorithm it issues a write operation trying to revise 
$ver_0$ and propose its own local value as the value to be decided. 
According to strong coverability only a single write operation 
$\act{cvr-write}(v,ver_0)(v, ver_1, chg)$ is going to succeed proposing its value, say $v$, 
and change the version of the register from $ver_0$ to some version $ver_1$. 
All the rest of the write operations will be of the form $\act{cvr-write}(v',ver_0)(v, ver_1, unchg)$
and thus will fail to change the value and version of the register. 
The write operation will return 
$(lcval,lcver, flag)=(v,ver_1, unchg)$ no matter what value they tried to propose, 
and each will be able to agree on value $lcval = v$ reaching this way agreement.
This discussion yields the following theorem.
%

\begin{theorem}
The construction in Figure \ref{fig:consensus} implements 
a consensus object.
\end{theorem}


Figure \ref{fig:strongtr} shows the implementation of a strongly coverable atomic 
register using a consensus object. For our implementation of consensus we assume that
the consensus oracle runs a separate instance of consensus on each version of the 
object. Thus, the oracle accepts as inputs the version we want to revise as well as 
the $\tup{v, ver'}$ tuple that consists of the value we propose. When that value 
is not specified, the oracle returns the tuple decided on the instance associated
with the given version. If no consensus was reached for a given version then 
the oracle returns the tuple $\tup{\bot, \bot}$. To generate a new version a 
process calls the function \act{generate-version(ver)}. This procedure produces 
a unique version larger than any previous version, 
each time is executed. A trivial implementation of this 
function is to append the given version with the unique 
id of the invoking process. 

\begin{figure}[!ht]
	
	\hrule\vspace{0.15cm}
	\begin{footnotesize}
	At each process $i\in\idSet$\\
	Local Variables: $lcver, ver_{new}\in Versions\text{ initially }ver_0$; $lcval \in Values$; $P\in Values \times Versions$
	\begin{multicols}{2}
	{\bf function} \act{cvr-write}($v,ver$)
	\begin{algorithmic}
		\State $ver_{new} \gets \act{generate-version}(ver)$ 
		\State $P \gets \act{propose}(ver, \tup{v, ver_{new}})$
		\If{ $P.ver == ver_{new}$ }
			\State $lcver \gets ver_{new}$
			\State return $\tup{P, chg}$
		\Else
			\While{ $P.ver \neq \bot$ }
				\State $\tup{lcval, lcver} \gets \tup{ P.val, P.ver}$
				\State $P \gets \act{propose}(lcver, \bot)$
			\EndWhile
		\State return $\tup{lcval, lcver, unchg}$
		\EndIf
	\end{algorithmic}
	
	{\bf function} \act{cvr-read}()
	\begin{algorithmic}
		\State $P \gets \act{propose}(\bot, lcver)$
		\While{ $P.ver \neq \bot$ }
			\State $\tup{lcval, lcver} \gets \tup{P.val, P.ver}$
			\State $P \gets \act{propose}(\bot, lcver)$
		\EndWhile
			
		\State return $\tup{lcval, lcver}$\\
	\end{algorithmic}
	\end{multicols}
	\end{footnotesize}
	\hrule\vspace{.4em}
	\caption{Strongly Coverable Atomic Registers using Consensus}\vspace{-1em}
	\label{fig:strongtr}
\end{figure}

\begin{theorem}
The construction in Figure \ref{fig:strongtr} implements 
a strongly coverable atomic register.
\end{theorem}

\begin{proof}
We show that the algorithm satisfies two properties: 
$(i)$ strong coverability and $(ii)$ atomicity. 

Strong coverability
requires that only a single write operation changes each 
version of the register. 
Let us assume to derive contradiction 
that there exists a version $ver$ of the object s.t. two operations
$\op_1 = \act{cvr-write}(v, ver)(v, ver_1, chg)$ and 
$\op_2 = \act{cvr-write}(v', ver)(v, ver_2, chg)$ both 
revise $ver$ leading to two potentially different versions 
$ver_1$ and $ver_2$. For this to be possible it means that 
$P.ver = ver_1$ for $\op_1$ and $P.ver = ver_2$ for $\op_2$.
$P$ however is the value decided by the consensus oracle. 
Since both $\op_1$ and $\op_2$ revise the same version $ver$ 
then they both invoked the consensus oracle on the same instance 
of the version $ver$. Since the consensus oracle reaches agreement 
on a single value then it must be the case that $P$ is the same 
for both $\op_1$ and $\op_2$, and hence $P.ver=ver_1=ver_2$. This
however contradicts our assumption. Thus, only a single write 
operation is able to modify each version and this preserves
strong coverability.  

Atomicity is trivially preserved by the write operations 
as they follow the total order imposed by the versions they change. 
Read operations are ordered in terms of the write operations 
since they invoke the consensus oracle until they reach the latest
version of the object. A read operation $\rd_1$ 
does not return an older 
value than a preceding read $\rd_2$, since $\rd_2$ would reach 
an earlier or at most the same version as $\rd_1$ before 
completing. Thus, $\rd_2$ 
will return the same or an older value as desired. Finally, 
a write operation that does not change the version of the 
register it must be ordered with respect to the rest of the 
read operations. Such write also discovers the latest accepted
version and thus, as before, it will 
return the same or a newer value than the one returned by a 
preceding read or unsuccessful write operation. 
\end{proof}


\newpage

\section{Supporting Large Files}
\label{appx:files}

Figure \ref{fig:ldr} depicts a modified version of the LDR algorithm
\cite{FL03}, that implements versioned large objects.

\begin{figure}[!h]
\hrule\vspace{0.15cm}
	\begin{footnotesize}		
			
		$\act{tr-write}(val, ver=maxtag)$
			\begin{itemize}[leftmargin=5mm]
				\item[] \act{get-metadata}: Send query request to \emph{directory servers} and wait for
				 $(tag,location)$ responses from a majority of them. Select the 
				$(tag,location)$ among the collected replies with the largest tag; 
				let $\tup{\tau,\srvSet}$ be this pair and the integer component of $\tau$ be $z$. Then:
				
				\item[] If $\tau \neq ver$	then do the following: 
				\begin{itemize}[leftmargin=5mm]
					\item[] \act{put-metadata}: Send $\tup{\tau,\srvSet}$ to the 
					\emph{directory servers} and wait for a majority of them to reply. Once those
					replies are received set $\tup{\tau_{new},\srvSet_{new}} = \tup{\tau, \srvSet}$.
					
					\item[] \act{get}: Send \emph{get object} request to $f+1$ \emph{replica servers}
					in $\srvSet$ for the $\tau$ version of the object and wait for a single server 
					to reply with $x$. Return $\tup{x, \tau, unchg}$. 
				\end{itemize}
				
				\item[] If $\tau = ver$ then do the following: 
				\begin{itemize}[leftmargin=5mm]
					\item[] \act{put}: 
					Create a new tag $\tau_{new} = \tup{z+1, wid}$ where $wid$ is the unique 
					identifier of the writer. Send $\tup{\tau_{new}, val}$ to $2f+1$ \emph{replica servers} 
					and wait for $f+1$ replies. Collect the identifiers of the servers that replied in 
					a set $\srvSet_{new}$.
					\item[] \act{put-metadata}: Send $\tup{\tau_{new},\srvSet_{new}}$ to all the \emph{directory servers}
					and wait for the majority of them to reply. Return $\tup{val, \tau_{new}, chg}$. 
				\end{itemize}				

			\end{itemize}
			
			$\act{tr-read}()$
			\begin{itemize}[leftmargin=5mm]
				\item[] \act{get-metadata}: Send query request to \emph{directory servers} and wait for
				$(tag,location)$ responses from a majority of them. Select the 
				$(tag,location)$ among the collected replies with the largest tag; 
				let $\tup{\tau,\srvSet}$ be this pair and the integer component of $\tau$ be $z$.
				
				\item[] \act{put-metadata}: Send $\tup{\tau,\srvSet}$ to the 
				\emph{directory servers} and wait for a majority of them to reply
					
				\item[] \act{get}: Send \emph{get object} request to $f+1$ \emph{replica servers}
				in $\srvSet$ for the $\tau$ version of the object and wait for a single server 
				to reply with $x$. Return $\tup{x, \tau}$. 					
			\end{itemize}
			
		\act{directory-server}
			\begin{itemize}[leftmargin=5mm]
				\item[] On receipt of \act{get-metadata} message: Send the tag-locations 
				pair $\tup{\tau_{s},\srvSet}$ stored locally.
				
				\item[] On receipt of \act{put-metadata} message: Let $\tup{\tau_m,\srvSet_m}$ be the tag-location
				pair enclosed in the received message and $\tup{\tau_s,\srvSet}$ the local pair on the server. 
				Compare the tags $\tau_m$ and $\tau_{s}$. If $\tau_{m} > \tau_{s}$ and $|\srvSet_m|\geq f+1$ then 
				store $\tup{\tau_m,\srvSet_m}$ locally.
			\end{itemize}
			
		\act{replica-server}
			\begin{itemize}[leftmargin=5mm]
				\item[] On receipt of \act{put} message: Add the $\tup{\tau_m,value}$ pair enclosed in the message to 
				the local set of available pairs and send an acknowledgement. 
				
				\item[] On receipt of \act{get} message: If the value associated with the requested tag is in the 
				set of pairs stored locally, respond with the value. Otherwise ignore the message.
			\end{itemize}
			
		\end{footnotesize}
		\hrule\vspace{.4em}
\caption{Operations of the modified LDR algorithm}
\label{fig:ldr}
\end{figure}

\end{document}